\documentclass{article}
\usepackage{amsthm}
\usepackage{latexsym}
\usepackage{amssymb}
\usepackage{authblk}

\newtheorem{theorem}{Theorem}
\newenvironment{remark}[1][Remark]{\begin{trivlist}
\item[\hskip \labelsep {\bfseries #1}]}{\end{trivlist}}
\date{}
\usepackage{geometry}
\newgeometry{tmargin=3cm, bmargin=3cm, lmargin=3cm, rmargin=3cm}

\author[1]{Grzegorz Rz\c{a}dkowski}
\author[2,3]{Wojciech Rz\c{a}dkowski}
\author[4]{Pawe\l~W\'{o}jcicki}

\affil[1]{Department of Finance and Risk Management, Warsaw University of 
Technology, Ludwika Narbutta 85\\ 
 00-999 Warsaw, Poland\\
email: g.rzadkowski@wz.pw.edu.pl}
\affil[2]{Faculty of Physics, University of Warsaw, Ho\.za 69\\
 00-681 Warsaw, Poland\\}
\affil[3]{ Faculty of Electronics and Information Technology, Warsaw University 
of  Technology, Nowowiejska 15/19\\
 00-665 Warsaw, Poland}
\affil[4]{Faculty of Mathematics and Computer Science, Warsaw University of 
Technology, Koszykowa 75 \\
 00-662 Warsaw, Poland\\
 }

\title{On some connections between the Gompertz function and~special numbers}

\begin{document}

\maketitle
\thispagestyle{empty}


\begin{abstract}
In the present paper we show that the Gompertz function, the Fisher--Tippett and 
the Gumbel probability distributions are related to both Stirling numbers of the 
second kind and Bernoulli numbers. Especially we prove for the Gumbel 
probability density function an analog of the Grosset--Veselov formula which 
connects 1-soliton solution of the KdV equation with Bernoulli numbers.  
\end{abstract}


\section{Introduction}
The differential equation defining the Gompertz function is as follows
\begin{equation}\label{2a}
	u'(t)=qu\log\frac{u_{max}}{u},\quad u(0)=u_{0}>0,
\end{equation}
where $t$ is time, $u=u(t)$ is an unknown function, $q, u_{max}$ are constants. 
 
Constant $u_{max}$ is called the saturation level. The integral curve $u(t)$ 
fulfilling condition $0<u(t)<u_{max}$ is known as the Gompertz function. \\ 
Equation (\ref{2a}) is the first order ordinary differential equation which is 
easily solved by the method of the separation of variables.
After solving it we get the following formula for the Gompertz function 
\begin{equation}\label{2b}
	u(t)=u_{max}e^{-ce^{-qt}},
\end{equation}
where constant $c$ appears in the integration process and is connected with the 
initial condition $\displaystyle u(0)=u_{0}=u_{max}e^{-c}$, therefore 
$\displaystyle c=\log\frac{u_{max}}{u_{0}}$. 

If $u_{max}=1 $ then (\ref{2b}) is the cumulative distribution function for the 
Fisher--Tippett distribution. A particular case of the Fisher--Tippett 
distribution (for $c=q=1$) is the Gumbel distribution  with the probability 
density function (pdf): $\displaystyle g(t)=e^{-e^{-t}}e^{-t}$.

The Gumbel distribution has some application in physics. For example, Itoh, 
Mahmoud and Takahashi~\cite{Itoh} consider a 
discrete model of wave propagation. They consider the wave propagation as 
motion of an ensemble of particles, each of them jumping forward with certain 
probability. The distribution of the quantity $\frac{L-k\ln k}{k}$, where $L$ 
is the wavelength and $k$ -- number of particles is proven to approach Gumbel 
distribution for very large number of particles. This result corresponds to 
propagation of uniform soliton.

This idea is further explored by Itoh and Mahmoud~\cite{Itoh2} as they study 
the Moran population model. An ensemble of $k$ gametes is considered; each of 
gametes can die and produce a new gamete of age 0. The probability of such 
evolution depends on gamete's age. Authors conclude that for large $k$ Gumbel 
distribution approximates $\frac{M-k\ln k}{k}$, where $M$ is the maximal age of 
the population.

Gumbel distribution is also of importance in geophysics. It was used to study 
extreme river floods~\cite{flood1,flood2,flood3}, seismicity in 
Eurasia~\cite{seismicity}, occurence of extreme meteorogical 
events~\cite{extreme} or the level of Aegan and Ionian seas~\cite{Greece}.

The Gompertz curve is also directly used for modelling different phenomena in 
biophysics, 
biology, medicine and others (see Stauffer et al~\cite{S}, Waliszewski and 
Konarski~\cite{WK}).

The paper is organized as follows.  In Sec. 2 we recall some properties of 
Stirling numbers of the second kind and Bernoulli numbers which we use later. In 
Sec. 3 we show our main results concerning the Gompertz function and Gumbel's 
distribution. The paper is concluded in Sec. 4.

\section{Properties of special numbers}
By $\displaystyle { n \brace k} $ we denote the Stirling number of the second 
kind (number of ways of partitioning a set of $n$ 
elements into $k$ nonempty subsets; see Graham et al.
\cite{GKP}). It is set  $\displaystyle{ n \brace 0}=0 $ for $n>0$, 
$\displaystyle { 0 \brace 0}=1 $,  $\displaystyle { n \brace k}=0 $ for $k>n$, 
or 
$k<0$.
Let us recall that the numbers fulfill 
{\setlength\arraycolsep{2pt}
\begin{eqnarray}
&&{ n \brace k}
=\frac{1}{k!}\sum_{j=0}^{k}(-1)^{k-j}{k \choose j}j^n=
\frac{1}{k!}\sum_{j=0}^{k}(-1)^{j}{k \choose j}(k-j)^{n},\label{1a}\\
&&	{ n+1 \brace k}=k{ n \brace k}+ { n \brace k-1},\label{1b}
\end{eqnarray}}
and appear in the Taylor expansion
\[\frac{(e^{w}-1)^{k}}{k!}=	\sum_{n=k}^{\infty }{ n \brace k}
\frac{w^{n}}{n!}.
\]
By $B_{n}$ ($n=0,1,2,\ldots $) we denote the $n$th Bernoulli number. The 
Bernoulli numbers have the following exponential  generating function (see 
Duren \cite{D})
\begin{equation}\label{1c}
B_0+B_{1}z + B_{2}\frac{z^{2}}{2!}+\cdots =\frac{z}{e^{z}-1},
\qquad |z|<2\pi.
\end{equation}
It is well known that $B_{n}$ vanishes for odd $n\ge 3$. The numbers are 
rational and they appear in relations such that
\[\sum_{k=1}^{m-1}k^{n}=\frac{1}{n+1}\sum_{j=0}^{n}{n+1 \choose j} 
B_{j}m^{n+1-j},\quad m,n\ge 1,
\]
or	
\[\sum_{k=1}^{\infty}\frac{1}{k^{2n}}=(-1)^{n+1}
\frac{2^{2n-1}\pi^{2n}}{(2n)!}B_{2n}, \quad n=1,2,\ldots
\]
The first few nonzero Bernoulli numbers are as follows
\[B_{0}=1,\;B_{1}=-\frac{1}{2},\;B_{2}=\frac{1}{6},\;B_{4}=-\frac{1}{30}
,\;B_{6}=\frac{1}{42},\;B_{8}=-\frac{1}{30},\;B_{10}=\frac{5}{66},\;B_{12}
=-\frac{691}{2730}.
\]
Since 
\[\frac{1}{e^{z}+1}=\frac{e^{z}+1-2}{e^{2z}-1}=\frac{1}{e^{z}-1}-\frac{2}{e^{2z}
-1}=\frac{1}{z}\cdot\frac{z}{e^{z}-1}-\frac{1}{z}\cdot\frac{2z}{e^{2z}-1}
\]
we infer from (\ref{1c}) that
{\setlength\arraycolsep{2pt}}
\begin{eqnarray}
\label{1d}
\frac{1}{e^{z}+1}&=&\frac{1}{z}\sum\limits_{n=0}^{\infty}B_{n}\frac{z^{n}}{n!}-
\frac{1}{z}\sum\limits_{n=0}^{\infty}B_{n}\frac{2^{n}z^{n}}{n!}
=\sum\limits_{n=1}^{\infty}\left(B_{n}\frac{z^{n-1}}{n!}-B_{n}\frac{2^{n}z^{n-1}
}{n!}\right)\nonumber\\
&=&\sum\limits_{n=0}^{\infty}\frac{B_{n+1}(1-2^{n+1})}{(n+1)!}\;z^{n},  
\end{eqnarray}
for $|z|<\pi$. \\

\section{Main results}

\begin{theorem}
If $u(t)$ is a solution of equation (\ref{2a}) then its $n$th derivative is 
given by the formula
\begin{equation}\label{2d}
	u^{(n)}(t)=q^{n}\sum\limits_{k=1}^{n}(-1)^{n-k}{ n \brace 
k}u\log^{k}\frac{u_{max}}{u}
\end{equation}
\end{theorem}
\begin{proof}
We will proceed by induction. For $n=1$ formula (\ref{2d}) transforms into 
(\ref{2a}), thus is true. Let us suppose that for a positive integer $n$ 
formula 
(\ref{2d}) holds. Then using the chain rule and identity (\ref{1b}) we get
\begin{eqnarray*}
\lefteqn{u^{(n+1)}(t)=q^{n}\frac{d}{dt}\sum\limits_{k=1}^{n}(-1)^{n-k}{ n 
\brace 
k}u\log^{k}\frac{u_{max}}{u}}
\\ && =q^{n}\sum\limits_{k=1}^{n}(-1)^{n-k}{ n \brace 
k}\left(\log^{k}\frac{u_{max}}{u}-k\log^{k-1}\frac{u_{max}}{u}\right)
qu\log\frac{u_{max}}{u}
\\ && =q^{n+1}\sum\limits_{k=1}^{n}(-1)^{n+1-k}{ n \brace 
k}\left(ku\log^{k}\frac{u_{max}}{u}-u\log^{k+1}\frac{u_{max}}{u}\right)\\
&&=q^{n+1}\bigg[(-1)^{n}u\log\frac{u_{max}}{u}+\sum\limits_{k=2}^{n}(-1)^{n+1-k}
\left(k{ n \brace k}+{ n \brace k-1}  \right)
u\log^{k}\frac{u_{max}}{u}\\
&&\hspace{3mm}+u\log^{n+1}\frac{u_{max}}{u}\bigg]=q^{n+1}\sum\limits_{k=1}^{n+1}
(-1)^{n+1-k}{ n+1 \brace k}u\log^{k}\frac{u_{max}}{u}
\end{eqnarray*}
\end{proof}
Let us denote by $P_{n}(u)$ the right hand side of (\ref{2d}) 
	\[P_{n}(u)=q^{n}\sum\limits_{k=1}^{n}(-1)^{n-k}{ n \brace 
k}u\log^{k}\frac{u_{max}}{u}, \qquad n=1, 2, 3,\ldots
\]
and $P_{0}(u)=u$. 
\begin{remark}
The Bell polynomials are defined as $\displaystyle 
B_{n}(x)=\sum\limits_{k=1}^{n}{ n \brace k}x^{k}$ (see Comtet \cite{C}). Thus 
function $P_{n}(u)$ can be expressed in terms of the Bell polynomials as follows
	\[P_{n}(u)=q^{n}\sum\limits_{k=1}^{n}(-1)^{n-k}{ n \brace 
k}u\log^{k}\frac{u_{max}}{u} =(-q)^{n}u\sum\limits_{k=1}^{n}{ n \brace 
k}\bigg(-\log\frac{u_{max}}{u}\bigg)^{k}=(-q)^{n}u B_{n}\bigg( 
\log\frac{u}{u_{max}}  \bigg)
\]
\end{remark}

Now we introduce the exponential generating function (e.g.f.) for 
$\{P_{n}(u)\}$ 
\begin{equation}\label{2e}
	G(u,z)=P_{0}(u) 
+P_{1}(u)z+P_{2}(u)\frac{z^{2}}{2!}+P_{3}(u)\frac{z^{3}}{3!}+\cdots
\end{equation}
In order to find a closed form formula for e.g.f. (\ref{2e}) observe that 
$G(u(t),z)$ is the Taylor series expansion of the function $u(t+z)$ at point 
$t$. Therefore using formula (\ref{2b}) we get
\begin{equation}
\label{2f}
G(u(t),z)=u(t+z)=u_{max}e^{-ce^{-q(t+z)}}=u_{max}\big(e^{-ce^{-qt}}\big)^{e^{-qz
}}.
\end{equation}
From (\ref{2f}) it follows that
\begin{equation}\label{2ff}
	G(u,z)=u_{max}\left(\frac{u}{u_{max}}\right)^{e^{-qz}}
\end{equation}
\begin{theorem}\label{th2}
For the Gumbel's pdf $\;\;\displaystyle g(t)=e^{-e^{-t}}e^{-t}$ the following 
formula holds
\begin{equation}
\label{2n}
\int_{-\infty}^{\infty}\bigg(\frac{d^{k-1}}{dt^{k-1}}(e^{-e^{-t}}e^{-t})\bigg)^{
2}dt=(-1)^{k}\frac{B_{2k}(1-2^{2k})}{2k},
\end{equation}
where $k=1,2,\ldots$.
\end{theorem}

\begin{proof}
Our aim is to compute integral $p_{n}=\int_{0}^{u_{max}}P_{n}(u)du$, 
$n=0,1,2,\ldots $. By (\ref{2ff}) the e.g.f. for $\{p_{n}\}$ is
\begin{eqnarray}
\lefteqn{\hspace{-25mm}p_{0} 
+p_{1}z+p_{2}\frac{z^{2}}{2!}+p_{3}\frac{z^{3}}{3!}+\cdots 
=\int_{0}^{u_{max}}G(u,z)du=
u_{max}\int_{0}^{u_{max}}\left(\frac{u}{u_{max}}\right)^{e^{-qz}}du \nonumber 
}\\
&&=u_{max}^{2}\cdot 
\frac{1}{e^{-qz}+1}\left(\frac{u}{u_{max}}\right)^{e^{-qz}+1}\bigg|_{0}^{u_{max}
}=
\frac{u_{max}^{2}}{e^{-qz}+1}.\label{2g}
\end{eqnarray}
Using formula (\ref{1d}) we obtain
\begin{equation}\label{2h}
\hspace{-0.5mm}	
\frac{1}{e^{-qz}+1}=\sum\limits_{n=0}^{\infty}\frac{B_{n+1}(1-2^{n+1})}{(n+1)!}
(-qz)^{n}=
	\sum\limits_{n=0}^{\infty}\frac{(-q)^{n}B_{n+1}(1-2^{n+1})}{n+1}\cdot 
\frac{z^{n}}{n!}
\end{equation}
and then comparing coefficients of $z^{n}/n! $ in  (\ref{2g}) and (\ref{2h}) we 
get
\begin{equation}\label{2j}
	p_{n}=\int_{0}^{u_{max}}P_{n}(u)du 
=\frac{(-q)^{n}B_{n+1}(1-2^{n+1})}{n+1}\cdot u_{max}^{2}
\end{equation}
Substituting in the integral(\ref{2j}) variable $t$, 
$u(t)=u_{max}e^{-ce^{-qt}}$ 
we get
\begin{equation}
\label{2k}
\int_{0}^{u_{max}}P_{n}(u)du=\int_{-\infty}^{\infty}u^{(n)}(t)u'(t)dt=\frac{
(-q)^{
n}B_{n+1}(1-2^{n+1})}{n+1}\cdot u_{max}^{2}
\end{equation}
If $n$ is an odd number $n=2k-1$ then integrating in (\ref{2k}) $k-1$ times by 
parts we have
	
\[(-1)^{k-1}\int_{-\infty}^{\infty}(u^{(k)}(t))^{2}dt=\frac{(-q)^{2k-1}B_{2k}
(1-2^{2k})}{2k}\cdot u_{max}^{2},
\]
or
\begin{equation}
\label{2m}
\int_{-\infty}^{\infty}(u^{(k)}(t))^{2}dt=(-1)^{k}\frac{q^{2k-1}B_{2k}(1-2^{2k})
}{2k}\cdot u_{max}^{2}
\end{equation}
Putting in (\ref{2m}) $u_{max}=1$, $q=1$, $c=1$ (in this case 
$u(t)=e^{-e^{-t}}, 
\; u'(t)=e^{-e^{-t}}e^{-t}$) we obtain
formula (\ref{2n}).
\end{proof}
\begin{remark}
For the Gumbel's pdf the formula (\ref{2n}) can be seen as an analog of the 
Grosset-Veselov formula 
\begin{equation}\label{int}
	B_{2k}=\frac{(-1)^{k-1}}{2^{2k+1}}\int_{-\infty}^{+\infty} 
\left(\frac{d^{k-1}}{dx^{k-1}}
	\frac{1}{\cosh^{2}x}\right)^{2}dx,
\end{equation}
which has been demonstrated in \cite{GV}. For another proofs of (\ref{int}) see 
\cite{Bo} and \cite{R}. Formula (\ref{int}) shows the connection between 
1-soliton solution of the KdV equation and the Bernoulli numbers.
\end{remark}
We may also prove some explicit formulae, expressing Bernoulli numbers in terms 
of Stirling numbers of the second kind or binomial numbers. In order to do this 
let us first observe that integrating by parts we get the following recurrence 
formula

\[\int_{0}^{u_{max}}u\log^{n}\frac{u_{max}}{u}\:du=\int_{0}^{u_{max}}\left(\frac
{u^{2}
}{2}\right)'\log^{n}\frac{u_{max}}{u}\:du=
\frac{n}{2}\int_{0}^{u_{max}}u\log^{n-1}\frac{u_{max}}{u}\:du
\]
which by $\int_{0}^{u_{max}}u\: du=u_{max}^{2}/2$ leads to 
\begin{equation}\label{2p}
\int_{0}^{u_{max}}u\log^{n}\frac{u_{max}}{u}\:du=\frac{n!}{2^{n+1}}\cdot 
u_{max}^{2}.
\end{equation}
In view of the definition of polynomial $P_{n}(u)$ and using (\ref{2p}) we can 
rewrite formula (\ref{2j}) as
\[  u_{max}^{2} q^{n}\sum\limits_{k=1}^{n}(-1)^{n-k}{ n \brace 
k}\frac{k!}{2^{k+1}}  =\frac{(-q)^{n}B_{n+1}(1-2^{n+1})}{n+1}\cdot u_{max}^{2},
\]
i.e.,
\begin{equation}\label{2r}
\sum\limits_{k=1}^{n}(-1)^{k}{ n \brace 
k}\frac{k!}{2^{k+1}}=\frac{B_{n+1}(1-2^{n+1})}{n+1}.
\end{equation}
Expressing in (\ref{2r}) the Stirling numbers by binomial numbers (\ref{1a}) we 
get 
\begin{equation}
\sum\limits_{k=1}^{n}\frac{1}{2^{k+1}}\sum_{j=0}^{k}(-1)^{k-j}{k 
\choose 
j}j^n=\frac{B_{n+1}(1-2^{n+1})}{n+1}.
\end{equation}

\section{Conclusions and further work}
We investigated some connections between the Gompertz function and special 
numbers (the Stirling numbers of the second kind and the Bernoulli numbers). We 
showed formula (\ref{2d}), involving Stirling numbers of the second kind, which 
expresses the $n$th derivative of the Gompertz function by this function itself. 
In a particular case ($u_{max}=1$) the Gompertz function is the Fisher-Tippett 
cumulative distribution function. A particular case of the Fisher-Tippett 
distribution (for $c=q=1$) is the Gumbel 
(standard) distribution. \\
Then we proved integral formula (\ref{2n}, Theorem~\ref{th2}), which shows a 
connection between the Gumbel pdf and Bernoulli numbers. In fact in the proof of 
Theorem~\ref{th2} such connection is proved slightly more generally, for the 
Fisher-Tippett pdf and the Gompertz function. Formula (\ref{2n}) can be seen as 
an analog of the Grosset-Veselov integral formula which connects 1-soliton 
solution of the KdV equation with Bernoulli numbers.\\
We get also in Sec. 3 some formulae expressing Bernoulli numbers in terms of 
Stirling numbers of the second kind or binomial numbers.\\
The approach used here could be applied for other functions or even for systems 
of nonlinear differential equations like the Lotka--Volterra model.

\section*{}

\end{document}